\documentclass[11pt]{article}

\usepackage[normalem]{ulem}

\usepackage{amsfonts,amsmath,amsthm,times,fullpage,amssymb,epsfig,graphicx}

\usepackage{enumerate}
\usepackage{soul}

\usepackage{hyperref} 

\hypersetup{colorlinks=false}

\usepackage[dvipsnames]{xcolor}

\usepackage{color}
\usepackage{dsfont}
\usepackage{wrapfig}
\usepackage{relsize}

\usepackage{url}

\usepackage[left=1.0in,top=1.0in,right=1.0in,bottom=1.0in,nohead,textheight=10in,footskip=0.3in]{geometry}

\usepackage[noend]{algpseudocode} 
\usepackage[nothing]{algorithm}  
\usepackage{caption}

\usepackage{comment}

\newtheorem{theorem}{Theorem}[section]
\newtheorem{lemma}[theorem]{Lemma}

\newtheorem{definition}[theorem]{Definition} 
\newtheorem{remark}{Remark}[section]

\newtheorem*{gLLL}{General LLL}

\newtheorem*{localalgo}{Local Algorithms}

\newcommand\ex{{\mathbb{E}}}

\newcommand{\beq}{\begin{equation}}
\newcommand{\eeq}{\end{equation}}

\newcounter{fooTH}
\newcounter{fooEQ}

\begin{document}

\title{Simple Local Computation Algorithms for  the General \\ Lov\'{a}sz Local Lemma}

\author{
Dimitris Achlioptas
\\ University of Athens
\\
{\small optas@di.uoa.gr}
\\
\and
Themis Gouleakis
\thanks{
Research supported by NSF Award Numbers CCF-1650733, CCF-1733808, CCF-1740751, and IIS-1741137.} 
\\ University of Southern California
\\
{\small tgoule@mit.edu}
\\
\and 
Fotis Iliopoulos
\thanks{This material is based upon work directly supported by the IAS Fund for Math and indirectly supported by the National Science Foundation Grant No. CCF-1900460. Any opinions, findings and conclusions or recommendations expressed in this material are those of the author(s) and do not necessarily reflect the views of the National Science Foundation. This work is also supported by the National Science Foundation Grant No. CCF-1815328.} \\ 
Institute for Advanced Study\\
{\small fotios@ias.edu}
}

\date{\empty}

\maketitle

\begin{abstract}
We consider the task of designing Local Computation Algorithms (LCA) for applications of the Lov\'{a}sz Local Lemma (LLL). LCA is a class of sublinear algorithms proposed by Rubinfeld et al.~\cite{Ronitt} that have received a lot of attention in  recent years. The LLL is an existential, sufficient condition for a collection of sets to have non-empty intersection (in applications, often, each set comprises all objects having a certain property). The ground-breaking algorithm of Moser and Tardos~\cite{MT} made the LLL fully constructive, following earlier results by Beck~\cite{beck_lll} and Alon~\cite{alon_lll} giving algorithms under significantly stronger LLL-like conditions. LCAs under those stronger conditions were given in~\cite{Ronitt}, where it was asked if the Moser-Tardos algorithm can be used to design LCAs under the standard LLL condition. The main contribution of this paper is to answer this question affirmatively. In fact, our techniques yield LCAs for settings beyond the standard LLL condition.
\end{abstract}

\thispagestyle{empty}

\newpage
\setcounter{page}{1}\maketitle

\newpage

\section{Introduction}


The Lov\'{a}sz Local Lemma (LLL)~\cite{LLL} is a powerful tool of probabilistic combinatorics for establishing the existence of objects satisfying certain properties (constraints). As a probability statement, it asserts that given a family of ``bad'' events, if each bad event is individually not very likely and, in addition, is independent of all but a small number of other bad events, then the probability of avoiding all bad events is strictly positive. Given a collection of constraints, one uses the LLL to prove the existence of an object satisfying all of them (a perfect object) by considering, for example, the uniform measure on all candidate objects and defining one bad event for each constraint (containing all candidate objects that violate the constraint).  
Making the LLL constructive was the subject of intensive research for over two decades, during which several constructive versions were developed~\cite{beck_lll,alon_lll,mike_stoc,Czumaj_lll,aravind_08}, but always under conditions stronger than those of the LLL. In a breakthrough work~\cite{Moser,MT},  Moser and Tardos made the LLL constructive for any \emph{product} probability measure (over explicitly presented variables). Specifically, they proved that whenever the LLL condition holds, their {\sl Resample} algorithm, which repeatedly selects \emph{any} occurring bad event and resamples all its variables  according to the  measure, quickly converges to a perfect object.

In this paper we consider the task of designing \emph{Local Computation Algorithms} (LCA) for applications of the LLL. This is a class of sublinear algorithms proposed by Rubinfeld et al. in~\cite{Ronitt} that has received a lot of attention in the recent years~\cite{alon2012space,ghaffari2017complexity, hassidim2016local,sparse_spanning,spanners, aviad,mansour2013local, reingold2016new}. For an instance $F$, a local computation algorithm should answer in an online fashion, for any index $i$, the $i$-th bit of one of the possibly many solutions of $F$, so that the answers given are consistent with some specific solution of $F$. As an example, given a  constraint satisfaction problem and a sequence of queries corresponding to variables of the problem, the algorithm should output a value assignment for each queried variable that agrees with some full assignment satisfying all constraints (assuming one exists).

The motivation behind the study of LCAs becomes apparent in the context of computations on massive data sets. In such a setting, inputs to and outputs from algorithms may be too large to handle within an acceptable amount of time. On the other hand, oftentimes only small portions of the output are required at any point in time by any specific user, in which case the use of a local computation algorithm is appropriate. We also note that LCAs can be seen as a generalization of several models such as local algorithms~\cite{suomela2013survey}, locally decodable codes~\cite{yekhanin2012locally} and local reconstruction algorithms e.g.,~\cite{AilonCCSL08,BhattacharyyaGJJRW12, BlumLR90,JhaR11,SacksS10}.

The algorithm we propose is simple and essentially  corresponds to running the Moser-Tardos algorithm with a specific strategy for choosing which occurring bad event to resample. As an example, assume we are given a constraint satisfaction problem and a set of queries (variables) $x_1, x_2, \ldots, x_q$. In this case, the algorithm first finds a satisfying assignment for the instance induced by the constraints within distance $r$ of $x_1$ in the dependency graph, and then outputs the current value of $x_1$. Then it considers variable $x_2$ and the instance of constraints within distance $r$ of it, then $x_3$ and so on and so forth. Our key observation is that if the constraints within a ball of radius $r$  around variable $x$ are all satisfied after some  step of the execution of the Moser-Tardos algorithm, then the probability that the algorithm needs to resample $x$ in some subsequent step is exponentially small in $r$. We use this fact to show that if the LLL condition is satisfied, then we can choose  $r$ appropriately to get a sublinear time algorithm that makes no errors with high probability.

\subsection{Related work in local computation algorithms}
The original paper of Rubinfeld et al.~\cite{Ronitt} as well as the follow-up work of Alon et al.~\cite{alon2012space} provide LCAs for several problems, including applications of the LLL to $k$-SAT and hypergraph $2$-coloring. The LCAs for LLL applications given in these works, though, are based on the earlier constructive versions of the LLL by Beck~\cite{beck_lll} and by Alon~\cite{alon_lll}, thus requiring significantly stronger conditions than the (standard) LLL condition. Indeed, it was left as a major open question in~\cite{Ronitt} whether the Moser-Tardos algorithm  can be used to design LCAs under the LLL condition. (Note also that, besides requiring stronger conditions, the algorithms of~\cite{beck_lll,alon_lll} are relatively involved compared to the Moser-Tardos algorithm.) We further discuss how our algorithm compares to the ones of~\cite{Ronitt,alon2012space}  in Section~\ref{Applications_SAT}.

Moreover, there is a recent line of research on LLL in the distributed $\mathsf{LOCAL}$ model~\cite{chang2018complexity,distributed, fischer2017sublogarithmic, ghaffari2018derandomizing} that often imply the existence of   LCAs for various problems. However, these works also require stronger conditions than the standard LLL condition and the resulting LCAs are significantly more sophisticated than the algorithm we propose in this paper.


\subsection{Our contributions}

Our main contribution is to make the LLL \emph{locally constructive}, i.e.,  to give a LCA under the LLL condition. Our techniques actually yield a LCA under more general recent conditions for the success of stochastic local search algorithms~\cite{AIJACM,AIK,HV} 
that go beyond the variable setting of Moser and Tardos. For simplicity of exposition, though, we focus our presentation on the variable setting of Moser and Tardos, as it captures the great majority of LLL applications, and discuss the more general settings later. That is, we focus on constraint satisfaction problems  $(\mathcal{X}, \mathcal{C})$, where $\mathcal{X}$ is a set of variables and $\mathcal{C}$ is a set of constraints over these variables. Given a product measure $\mu$ over $\mathcal{X}$, the LLL condition is said to be satisfied with $\epsilon$-slack for the family of bad events induced by $\mathcal{C}$, if the ``badness'' of each bad event is bounded by $1-\epsilon$ (see Section~\ref{LLL_sec}). Given an instance $(\mathcal{X}, \mathcal{C})$, we assume that each constraint entails at most $k = O( \mathrm{polylog}|\mathcal{X}|) $ variables, and each variable is entailed by at most $d = O( \mathrm{polylog}|\mathcal{X}|) $ constraints. Finally, a $(t,s,\delta)$-LCA responds to each query in time $t$,  using memory $s$, and makes no error with probability at least $1-\delta$. An informal version of our main result can thus be stated as follows. 
\begin{theorem}[Informal Statement]\label{main_intro}
If $( \mathcal{X}, \mathcal{C}, \mu)$  satisfies the LLL conditions with $\epsilon$-slack, then there exists an $( n^{\beta}, O(n) , n^{-\gamma}  )$-LCA for $(\mathcal{X}, \mathcal{C})$, for every $\beta, \gamma >0$  such that $ (1  + \gamma)/\beta < \log(1/(1-\epsilon))/ \log (k d)$.
\end{theorem}

Theorem~\ref{main_intro} gives a trade-off between the running time (per query) and the probability of error, while establishing that both decrease with the slack in the LLL conditions. Moreover, as we will see, if we know beforehand  the total number of queries to our algorithm, then the condition of Theorem~\ref{main_intro} can be significantly improved. (We stress that the latter is a feature of our results  which only adds flexibility to the original definition of LCAs and  does not impose any restrictions, as the user can always choose to not introduce a limitation on the number of queries. However, when dealing with large instances such limitations are natural and/or even unavoidable.)
 
 Using our general results we design LCAs for the following problems, chosen to highlight different features of our results. As we will see formally in Section~\ref{LCA_prel}, our results  apply to constraint satisfaction problems of large size, i.e., we assume that the number of variables is sufficiently large. This  mild assumption is essentially inherent in the model of local computation algorithms.

\subsubsection{$k$-SAT}\label{Applications_SAT}

Gebauer, Szab\'{o} and Tardos~\cite{2ke} used the  LLL to prove that any $k$-CNF formula where every variable appears in at most $d$ clauses is satisfiable if $d   (k+1)  \le  2^{k+1}/\mathrm{ e}$ and, moreover, that this is asymptotically tight in $k$. We show the following.
\begin{theorem}\label{SAT}
Let $\phi$ be a $k$-CNF formula on $n$ variables with $m$ clauses where every variable appears in at most $d$ clauses. 
\begin{enumerate}[(a)]

\item Suppose that  $\left[ d (k+1) \right]^{1+\eta}  \le 2^{k+1}/ \mathrm{e}$, for some constant $\eta >0$. For every  $\alpha, \beta, \gamma > 0 $ such that $ (\alpha + \gamma)/\beta < \eta$,  there exists a $(n^{\beta}, O(n^{ \min \{1, \alpha  +\beta \}}   ), n^{-\gamma})$-LCA for $\phi$ that answers up to $n^\alpha$ queries.

\item Suppose that  $d (k+1) \le (1-\epsilon) 2^{k+1}/\mathrm{e}$, for some constant $\epsilon >0$. Then, for every  $\beta , c > 0$, there exists a $( n^{\beta}, n^{\beta} \log^{c}(n), \log^{-c}(n) )$-LCA for $\phi$ that answers up to $\log^c (n)$ queries.

\end{enumerate}
\end{theorem}

For comparison, the work of Rubinfeld et al.~\cite{Ronitt} gave a LCA for $k$-CNF formulas only when there exist $k_1, k_2, k_3$ such that $k_1 + k_2 +k_3 =  k$ and
\begin{eqnarray*}
8d (d-1)^3 (d+1) & < & 2^{k_1} \\
8d (d-1)^3 (d+1) & < & 2^{k_2}  \\
\mathrm{e}(d+1) & < & 2^{k_3}\enspace.
\end{eqnarray*}
Notably, the LCA of~\cite{Ronitt} is logarithmic in time and space~\cite{alon2012space}. Unfortunately, the techniques  of~\cite{alon2012space} that allow for space-efficient local algorithms are tailored to the LLL-algorithm of Alon~\cite{alon_lll} and do not appear to be compatible with our results.

More specifically, Alon et al.~\cite{alon2012space} are able to exploit a technique introduced in~\cite{nguyen2008constant} that considers a random permutation of the input and feeds it to the algorithm in that order. In this way, they can use a pseudo-random generator in order to encode that permutation using logarithmic space. However, the successful application of this technique crucially relies on the fact that the algorithm in~\cite{alon_lll}, which is being simulated, can afford to sample each variable exactly once during the execution. (An additional  assumption, which we do not make in this paper, is that  each variable should be contained in a constant number of clauses.)  On the contrary, the Moser-Tardos algorithm works for the more general LLL conditions at the expense of the aforementioned property, which no longer holds. That is, it needs an explicit assignment of all variables at every point during the execution in order to know the set of currently violated clauses, while the algorithm in~\cite{alon_lll} can work only with partial value assignments until the very end of its execution, since each variable is assigned a value once. Therefore, a simple permutation of the input cannot capture the entire resampling sequence of the Moser-Tardos algorithm, which potentially involves multiple resamplings of each variable. Also, the constraint that only violated clauses are resampled makes certain resampling sequences invalid, and this even depends on the values sampled so far at any point of the execution, which is not the case in Alon's algorithm~\cite{alon_lll}.

\subsubsection{Coloring Graphs}\label{Applications_Coloring}

In graph vertex coloring one is given a graph $G(V,E)$ and the goal is to find a mapping of $V$ to a set of $q$ colors so that no edge in $E$ is monochromatic. The \emph{chromatic number}, $\chi(G)$, of $G$ is the smallest integer for which this is possible. Trivially, if the maximum degree of $G$ is $\Delta$, then $\chi(G) \le \Delta+1$. Molloy and Reed~\cite{molloy1997bound} proved that this can be significantly improved for graphs where the neighborhood of every vertex is bounded away from being a clique.

\begin{theorem}[\cite{molloy1997bound}]\label{seq_theorem}
There exists $\Delta_0$ such that if $G$ has maximum degree $\Delta > \Delta_0$ and the neighborhood of every vertex of $G$ contains at most $\binom{\Delta}{2} - B$ edges, where $B \ge \Delta  \log^4 \Delta$, then $\chi(G) \le \Delta +1 - B/(\mathrm{e}^6 \Delta)$.
\end{theorem}

Theorem~\ref{seq_theorem} is a sophisticated application of the LLL. Our results imply local algorithms for finding the colorings promised by Theorem~\ref{seq_theorem} that exhibit no trade-off between speed and accuracy, in the sense that for large enough $n$ both constants $\beta, \gamma$, below, can be made arbitrarily small.

\begin{theorem}\label{lca_theorem}
 Let $G$ be any graph on $n$ vertices, $m$ edges, and maximum degree $\Delta$  satisfying the conditions of Theorem~\ref{seq_theorem}. For every 
$\beta, \gamma >0 $ there exists a $(n^{\beta}, O(n) , n^{-\gamma})$-local algorithm for coloring $G$ using $\Delta +1 - B/(\mathrm{e}^6 \Delta)$ colors.
\end{theorem}

\subsubsection{Non-Uniform Hypergraph Coloring}

Our results can also handle applications of the LLL in non-uniform settings, i.e., where the probabilities of bad events may vary significantly. For example, it is known that a hypergraph  $\mathcal{H}$ with minimum edge size at least $3$ where every vertex lies in at most $\Delta_i$  edges of size $i$ is $2$-colorable, if $\sum_i \Delta_i 2^{-i/2 } \le  \frac{1}{ 6 \sqrt{2} }$ (see Theorem 19.2 in~\cite{mike_book}).

Using our main theorem  we can design a local algorithm for this problem when the number of queries is polylogarithmic. (Our  main result,  as well as extensions of the techniques in~\cite{Ronitt}, can be applied to give local algorithms with no restriction on the number of queries, but under significantly stronger assumptions for the $\Delta_i$. In particular, in these cases the fact that constraints corresponding to large hyperedges are ``easier'' to fix cannot be captured.) 


\begin{theorem}\label{non_uniform_hyper}
Fix $\epsilon > 0$ arbitarily small and $D>0$ arbitrarily large. Let $\mathcal{H}_{\epsilon, D}$ be the set of hypergraphs with minimum edge size at least $3$, where each vertex lies in at most $\Delta_i \le D$ edges of size $i$ such that  
\begin{align}\label{cond}
\sum_{i \ge 3} \Delta_i 2^{-i/2 } \le  \frac{1 - \epsilon}{ 6 \sqrt{2} } \enspace .
\end{align}
For every $\beta , c > 0$  there exists a $(  n^{\beta} , O(n^{\beta } \log^{c} n ), \log^{-c} n )$-LCA for $2$-coloring hypergraphs in $\mathcal{H}_{\epsilon, D}$ that answers up to $\log^c (n)$ queries.
 \end{theorem}

\section{Background}

\subsection{The Lov\'{a}sz Local Lemma}\label{LLL_sec}

To prove that a set of objects $\Omega$ contains at least one element satisfying a collection of constraints, we introduce a probability measure $\mu$ on $\Omega$, thus turning the objects violating each constraint into a bad event.

\begin{gLLL}
Let $(\Omega,\mu)$ be a probability space and $\mathcal{A} = \{A_1, A_2,\ldots,A_m\}$ be a set of $m$ (bad) events. For each $i \in [m]$, let $D(i) \subseteq	 [m] \setminus \{i\}$ be such that $\mu(A_i \mid \cap_{j \in S} \overline{A_j}) = \mu(A_i)$ for every $S \subseteq [m] \setminus (D(i) \cup \{i\})$. If there exist positive real numbers $\{\psi_i\}_{i=1}^m$ such that for all $i \in [m]$,
\begin{equation}\label{eq:LLL}
\frac{\mu(A_i)}{\psi_i}  \sum_{ S \subseteq   D(i) \cup \{i \} }  \prod_{j \in S} \psi_j \le 1 \enspace , 
\end{equation}
then the probability that none of the events in $\mathcal{A}$ occurs is at least $\prod_{i=1}^m 1/(1+\psi_i) > 0$. 
\end{gLLL}
\begin{remark}\label{general}
Condition~\eqref{eq:LLL} above is equivalent to the more well-known form $\mu(A_i) \le x_i \prod_{j \in D(i)} (1-x_j)$, 
where $x_i = \psi_i/(1+\psi_i)$. As we will see, formulation~\eqref{eq:LLL} facilitates refinements. To see the equivalence, notice that since $x_i =0$ is uninteresting, we may assume $x_i \in (0,1)$. Taking $\psi_i > 0$, 
setting $x_i = \psi_i/(1+\psi_i)  \in (0,1)$, and simplifying, the condition becomes $\mu(A_i) \prod_{j \in \{i \} \cup D(i) } (1+\psi_j) \le \psi_i$. Opening up the product yields~\eqref{eq:LLL}. 
\end{remark}

\begin{definition}
We say that the general LLL condition holds with $\epsilon$-slack if the righthand side of~\eqref{eq:LLL} is bounded by $1-\epsilon$ \, for every $i\in [m]$.
\end{definition}

Let  $G$ be the digraph over the vertex set $[m]$  having an arc from each $i \in [m]$ to each element of $D(i)\cup \{i \}$. We call such a graph a \emph{dependency} graph. Therefore, at a high level, the LLL states that if there exists a sparse dependency graph and each bad event is not too likely, then we can avoid all bad events with positive probability.

\subsection{Local Computation Algorithms}\label{LCA_prel}

\begin{definition}\label{ougada}
For any input $x$, define the set $F(x) = \{ y : y \text{ is a valid solution for input }x  \}$. The \emph{search} problem, given $x$, is to find any $y \in F(x)$. We use $\ell = |x|$ to denote the length of the input.

\end{definition}

Our definition of LCA algorithms is almost identical to the one of~\cite{Ronitt}, the only difference being that it is more flexible in the sense that it also takes as a parameter the  number of queries to the algorithm. 
\begin{localalgo}
Let $F(x)$ be as in Definition~\ref{ougada}. A $(
q
, t
, s
, \delta
)
$-local computation algorithm $\mathcal{A}$ is a (randomized) algorithm which satisfies the following: $\mathcal{A}$ receives a sequence $i_1, i_2, \ldots$ of up to $q(\ell)$ queries one by one; upon receiving each query $i_j$ it produces an output ${o_j}$; with probability at least $1-\delta(\ell)$, there exists $y \in F(x)$ such that $o_j = y_j$ for every $j$. $\mathcal{A}$ has access to a random tape and local computation memory on which it can perform current computations, as well as store and retrieve information from previous computations. We assume that the input $x$, the local computation tape and any random bits used are all presented in the RAM world model, i.e., $\mathcal{A}$ is given
the ability to access a word of any of these in one step. The running time of $\mathcal{A}$ on any query is at most $t(\ell)$, which is sublinear in $\ell$, and the local computation memory of $\mathcal{A}$ is at most $s(\ell)$.  Unless stated otherwise, we always assume that that the error parameter $\delta(\ell)$ is at most some constant, say, $\frac{1}{3}$. We say that $\mathcal{A}$ is a \emph{strongly local computation algorithm} if both $t(\ell), s(\ell)$ are upper bounded by $\mathrm{\log}^c \ell$ for some constant $c$.
\end{localalgo}

As we have already mentioned, in this paper we will be  interested  in local computation algorithms for constraint satisfaction  problems $(\mathcal{X}, \mathcal{C}) $, where $\mathcal{X}$ is a set of variables and $\mathcal{C}$ is a set of constraints over these variables.   
To simplify the statement of our results, whenever we say there exists a $(q,t,s,\delta)$-local computation algorithm for ($\mathcal{X}, \mathcal{C})$ we mean that there exists $n_0$ 
and an algorithm $\mathcal{A}$ such that $\mathcal{A}$ is a $(q,t,s,\delta)$-local computation algorithm when the input is restricted to instances of $(\mathcal{X}, \mathcal{C})$ such that $|\mathcal{X} | \ge n_0$. 
In other words, our results apply to constraint satisfaction problems of large size.

\section{Statement of Results}\label{Statement}

For simplicity, we will present our results and techniques for the general LLL in the \emph{variable setting}, i.e., the setting considered by  Moser and Tardos~\cite{MT}. In Section~\ref{improvedLLL} of the Appendix  we discuss how our techniques can be adapted to capture improved LLL criteria and generalized to  settings beyond the one of~\cite{MT}.

\paragraph*{The Setting.}
Let $\mathcal{X} = \{x_1, x_2, \ldots, x_n \}$ be a set of variables  with domains $D_1, \ldots, D_n$.  We define $\Omega = \prod_{i = 1}^n  D_i$ to be the set of possible value assignments for the variables of $\mathcal{X}$, and we sometimes refer to its elements as \emph{states}. We also consider a set of \emph{constraints}  $\mathcal{C}= \{ c_1, c_2, \ldots, c_m \}  $.  Each  constraint $c_i$ is  associated with a set of variables $\mathrm{var}(i) \subseteq \mathcal{X} $ and corresponds to a set  of forbidden value assignments for these variables, i.e., that \emph{violate} the constraint.

We consider an arbitrary product probability measure $\mu$ over the variables of $\mathcal{X}$ along with the family of bad events $\mathcal{A} = \{ A_1, \ldots, A_m \}$, where  $A_i$ corresponds to the states in $\Omega$ that violate  $c_i$. The \emph{dependency graph} $G = G(V,E)  $ related to $(\Omega, \mu, \mathcal{A})$ is the graph with vertex set  $V = [m]$ and edge set $E = \{ (i,j):  \mathrm{var}(i) \cap \mathrm{var}(j)  \ne \emptyset \}$. (Notice that since this dependence relationship is always symmetric, we have a graph instead of a digraph.) The \emph{neighborhood} of an event $A_i$  is defined as  $D(i)  = \{j: (i,j) \in E \}$ and notice that $A_i$ is mutually independent of $\mathcal{A} \setminus \left( D(i) \cup \{ i\} \right)$. Finally, for $i, j \in [m]$ we denote by $\mathrm{dist}(i, j)$ the length of a shortest path between $i$ and $j$ in $G$.

\paragraph*{Assumptions.}
 
We will make computational assumptions similar  to~\cite{Ronitt} (but less restrictive). For a variable $x$, we let $N(x)$ denote the set of constraints that contain $x$ and define $d = \max_{x \in \mathcal{X} } N(x)$. We further define an $n \times m$ incidence matrix $\mathcal{M}$ such that, for any variable $x$ and constraint $c$, $\mathcal{M}_{x,c} = 1$  if $c \in N(x)$ and $\mathcal{M}_{x,c} = 0$, otherwise. The input constraint satisfaction problem $(\mathcal{X}, \mathcal{C})$ will be represented by its variable-constraint incidence matrix $\mathcal{M}$. Let $k = \max_{i \in [m ] } |\mathrm{var}(i)|$ denote the maximum number of variables associated with a constraint. We will also assume that $d,k \in O( \log^c (n) )$ for some constant $c \ge0$,  which means that matrix $\mathcal{M}$ is necessarily very sparse. Therefore, we also assume that the matrix $\mathcal{M}$ is implemented via linked lists for each row (i.e., variable $x$) and each column (i.e., constraint $c$) and that 
\[
\max_{i \in [m ] } \psi_i = O(n^{\lambda}) \enspace 
\] 
for some constant $\lambda > 0$.
(Here the set of parameters $\{\psi_i \}_{i=1}^{m} $ is the one used in the LLL condition~\eqref{eq:LLL}. We note that in most applications $\max_{i \in [m] } \psi_i =O(1)$.) We can now state our main result precisely.
\begin{theorem}\label{main}
Assume that $( \mathcal{X}, \mathcal{C}, \mu)$  satisfies the Lov\'{a}sz Local Lemma conditions with  $\epsilon$-slack and define $\zeta = \zeta (\epsilon, k,d) =   \log (1/( 1- \epsilon) )/ \log (k d)  $. Let $\alpha, \beta, \gamma  >0 $ be constants such that $  \beta \zeta  > \alpha  + \gamma  + \lambda$. Then there exists a  $(n^{\alpha}, n^{\beta}, O(n^{ \min \{ 1, \alpha + \beta \} } ) , n^{-\gamma}  )$-local computation algorithm for $(\mathcal{X},\mathcal{C})$.
\end{theorem}

\begin{remark}\label{restricted_version}
If  the number of queries is $O( \mathrm{polylog} (n) ) )$, the probability of error  is $\Omega \left(  \frac{1}{ \mathrm{polylog} (n)  } \right)$, and $k,d =O(1)$, then if the LLL conditions hold with $\epsilon$-slack for some fixed constant $\epsilon > 0$, then for any arbitrarily small constant $\beta  >0$ there exists a LCA that takes $n^{\beta}$ time per query and uses $O(n^{\beta} \mathrm{polylog}(n) )$ space (for all sufficiently large $n$).
\end{remark}

\section{Our Algorithm}\label{the_algorithm}
In this section we describe our algorithm formally  as well as the main idea  behind its analysis.

To describe our algorithm, we first recall the algorithm of Moser and Tardos as well as a couple useful facts about its performance.

\begin{algorithm}
\begin{algorithmic}[1]  
\Procedure{RESAMPLE}{$\mu, \mathcal{C},\mathcal{X}$}

\State Sample all variables in $\mathcal{X}$ according to $\mu$

\While{violated constraints exist} 
\State Pick an arbitrary violated constraint $c_i$
\State (Re)sample every variable in $\mathrm{var}(i)$ according to $\mu$
\EndWhile

\EndProcedure
\end{algorithmic}
\end{algorithm}

Notice  that the most expensive operation of the Moser-Tardos algorithm is searching for constraints which are currently violated. In~\cite{spencer_hay}, a simple optimization is suggested to reduce this cost, which will be helpful to us as well.  The idea is to keep a stack which, at every step, contains all the currently violated constraints. To do that, initially, we go over all the constraints and add the violated ones into the stack.  Then, each time we resample a constraint $c$,  in order to update the stack, we are only required to check the constraints that share variables with $c$ to determine whether they became violated, in which case, we add them to the stack. The main benefit of maintaining  this data structure is that we avoid going over the whole set of constraints at each step. In particular,   using this method, we only have to put a $O(kd)$ amount of work after each resampling. This method is usually referred to as Depth-First MT.

In the following, when we say ``apply the Depth-First MT algorithm for at most $t$ steps", we mean that we apply the Resample algorithm above for at most $t$ steps, \emph{without} performing the initial sampling of the variables of $\mathcal{X}$ (all relevant variables will have been assigned values by other means).

For $i \in [m]$ and $r \ge 0$, let $\mathrm{Ball}(i,r) = \{ j \in [m]:  \mathrm{dist}(i,j)  \le r \} $ be the elements of $[m]$ whose distance to $i$ in $G$ is at most $r$. Furthermore, for a variable $x$ we denote by $\mathcal{I}(x,r)$ the sub-problem of $(\mathcal{X}, \mathcal{C})$ induced by the constraints  in  $ \bigcup_{c_i \ni x} \mathrm{Ball}(i,r) $ and the variables they contain. Notice that if $(\mathcal{X}, \mathcal{C})$ satisfies the LLL conditions, then $\mathcal{I}(x,r)$ does as well for any $x$ and $r$.  We are now ready to describe our meta-algorithm, that takes as input $q,t,\delta$ and $\epsilon$, i.e., the number of queries, the desired upper bounds on  the running time per query, the probability of error, and the slack, respectively.  For the sake of brevity, we slightly abuse notation and for $i \in [q]$ denote by $x_i$ the variable of the $i$-th query. 


\begin{algorithm}
\begin{algorithmic}[1]  
\Procedure{Respond to  Queries}{$q,t,\delta,\epsilon$}

\State $\eta \leftarrow \max_{x \in \mathcal{X} } \sum_{ c_j \ni x} \psi_j $.
\State $r \leftarrow \log  (q \eta/( \delta  - q/n^2 ) )/\log(1/( 1-\epsilon) ) 
 $
 \State  $ S \leftarrow \emptyset$
 \For { $i =1$ to $q$ }

 \State Resample each variable in $ \mathcal{I}(x_i, r) \setminus S $ \Comment $x_i, i \in [q]$, is the $i$-th query. \label{initial_resamplings}
 \State $S  \leftarrow S \cup  \mathcal{I}(x_i, r)  $
\State Apply the Depth-First MT algorithm to $\mathcal{I}(x_i,r)$ \label{run_MT} for at most $t$ steps
\If{a satisfying assignment for $\mathcal{I}(x_i,r)$ is found} 
 \State Output the value of $x_{i}$
\Else   
\State Abort \label{abort}
\EndIf
\EndFor
\EndProcedure
\end{algorithmic}
\end{algorithm}

The main idea behind our algorithm comes from the following property of the Moser-Tardos algorithm. Assume that in an execution of the Moser-Tardos algorithm, in the current step, every constraint in a ball of radius $r$ around variable $x$ is satisfied. We prove that the probability that the algorithm will have to resample $x$ in a later step drops exponentially fast with $r$. In other words, for large enough $r$, the current value of $x$ is a good guess for the value of $x$ in the final output. To exploit this fact, we use that in the Moser-Tardos algorithm the strategy for choosing which violated constraint to resample can be arbitrary, so that we get an LCA as follows: upon receiving query (variable) $x_i$, our algorithm tries to create a large ball of satisfied constraints around $x_i$, by executing the Moser-Tardos algorithm with a strategy prioritizing the constraints in the ball. Naturally, then the radius of the ball governs the trade-off between speed and accuracy.

\section{Proof of Theorem~\ref{main}}

In this section we present the proof of Theorem~\ref{main}.  Clearly, the  running time of our algorithm on any query is at most $t$. Further,  the local computation memory it requires is dictated by the number of variables it resamples (since it has to store the ``current'' value of every such variable), and the space  required for the stack in the application of the Depth-First MT. The former is at most linear while the latter is sublinear.  Therefore, we get  a $O(n)$ bound overall. (As it will become clear later, when the number of queries is limited, i.e., when $\alpha  < 1- \beta$, then the memory required is $O(n^{\alpha + \beta})$, i.e., sublinear.)

In the rest of the proof we will focus on bounding the probability that our algorithm makes an error.

Observe that Line~\ref{initial_resamplings} allows us to see the execution of our algorithm  as a prefix of a \emph{complete} execution of the Moser-Tardos algorithm from a random initial state. The probability that our algorithm makes an error is bounded by the sum of (i) the probability that our algorithm ever aborts in Line~\ref{abort}; (ii) the probability that the complete execution of the Moser-Tardos algorithm resamples a (queried) variable after our algorithm has returned its response for it. We start by bounding the former, since it's a more straightforward task.

\subsection{Bounding the Running Time as a Function of the Radius}\label{running_time_sec}

To bound the probability that our algorithm aborts in Line~\ref{abort}  we will use Theorem~\ref{run_time_MT} below, a direct corollary of the main result in~\cite{AIK}, bounding the running time of the Depth-First MT algorithm from an \emph{arbitrary} initial state. Let
\[
\xi = \max_{i \in [m] } \log (1 + \psi_i) \enspace .
\]


\begin{theorem}\label{run_time_MT}
If the LLL conditions hold with $\epsilon$ slack, then the probability that the MT algorithm starting at an arbitrary initial state has not terminated after $(n +    m \xi)/\log( 1/(1-\epsilon) )+ s$ steps is at most $(1-\epsilon)^s$. 
\end{theorem}

There are two  reasons why we need to use Theorem~\ref{run_time_MT} instead of the original running time bound of Moser and Tardos~\cite{MT}. The first and most important one, is that the original bound assumes that the initial state of the algorithm is selected according to the product measure $\mu$. However, when we run the 
MT algorithm in response to a query for variable $x_i$, some of the variables of $\mathcal{I}(x_i,r)$ may have been resampled multiple times in earlier executions of the for loop and, thus, be correlated with each other.  The second reason is that Theorem~\ref{run_time_MT} exploits the slack in the LLL conditions to ensure that the algorithm terminates fast with high probability and not just in expectation. 

{We are now ready to give a tail-bound for the running time of our algorithm on a single query, as a function of the radius $r$. Recall that each constraint contains at most $k$ 
variables and that each variable is contained in at most $d$ constraints and $k,d$ are at most polylogarithmic. We use $\widetilde{O}(\cdot)$ notation to hide poly-logarithmic factors in $n,m$.
\begin{lemma}\label{running_time}
Let $T_0 = (kd)^r (\xi/ \log(1/(1-\epsilon))$. Step~\ref{run_MT} takes more than $\widetilde{O}(T_0+s)$ time with probability at most 
$(1 - \epsilon)^s$.
\end{lemma}
\begin{proof}
Let us first derive an upper bound $B$ on the number of constraints (and of variables) in $\mathrm{Ball}(i, r)$. Since the maximum degree of the dependency graph is at most $kd$ and the subgraph that maximizes the number of constraints inside $\mathrm{Ball}(i,r)$ is the full $kd$-ary tree of depth $r$, we see that $| \mathrm{Ball}(i, r) | = O(  (kd)^{r+1} ) = \widetilde{O}( (kd)^r)$, since $k,d$ are at most poly-logarithmic. Thus, we can assume that $B = \widetilde{O}( (kd)^r)$.

The running of our algorithm on query $x_i$ consists  of computing the sub-problem $\mathcal{I}(x_i,r)$ and then applying Depth-First MT to it. By ``computing the sub-problem $\mathcal{I}(x_i,r)$" we mean creating an incidence matrix $\mathcal{M}_{i,r}$ that corresponds to the subgraph of the dependency graph associated with $\mathcal{I}(x_i,r)$, represented similarly to $\mathcal{M}$ via linked lists. To perform this task we can do  a Breadth First Search starting from a node $j$ such that $c_j \ni v_i$  for depth $r$. This takes $\widetilde{O}((kd)^{r} )$ time,  since we can find the neighbors of a constraint in the dependency graph in poly-logarithmic time   and the subgraph of the dependency graph that corresponds to $\mathcal{I}(x_i,r)$  has at most $Bkd = \widetilde{O}((kd)^{r})$ edges. 

For the application of Depth-First MT to $\mathcal{I}(x_i,r)$, Theorem~\ref{run_time_MT} asserts that if $
T_0 = (B + B \xi)/\log(1/(1-\epsilon))$,
then the probability that a satisfying assignment is not found after $T_0+s$ resamplings is at most $(1-\epsilon)^s$. 
Recalling that $B = \widetilde{O}( (kd)^r)$, that the amount of work per resampling is $O(kd)$, and that both $k$ and $d$ are polylogarithmic and adding the bound above for formulating each subproblem, concludes the proof.
\end{proof}

\subsection{Bounding the Probability of Revising a Variable as a Function of the Radius}\label{failure_probability}

To bound the probability of error of our algorithm we first need to recall a key element of the analysis of~\cite{MT}.
\subsubsection{Witness Trees}\label{WTL}

We denote by $\Sigma = \sigma_1 \xrightarrow{w_1}  \sigma_2 \xrightarrow{w_2} \sigma_3 \xrightarrow{w_3} \ldots $  the random variable that equals the \emph{trajectory} of an execution of the Moser-Tardos algorithm, where,  for each $i \ge 1$,
 $\sigma_i \in \Omega$  denotes the $i$-th state of the trajectory and $w_i \in [m]$  the index of the bad event resampled.  We  also call the random variable $W(\Sigma) = ( w_1, w_2, \ldots  ) $ the  \emph{witness sequence} of $\Sigma$. 
 
We first  recall the definition of witness trees from~\cite{MT}, while slightly reformulating to fit our setting. A witness tree $\tau  = ( T, \ell_{T} )$ is a finite rooted, unordered,  tree $T$ along with a labelling $\ell_T: V(T) \rightarrow [m] $ of its vertices with indices of bad events such that the children of a vertex $v \in V(T)  $ receive labels from $D( \ell(v) ) \cup \{ \ell(v) \} $. To lighten notation, we will sometimes write $(v)$ to denote $\ell(v)$ and $V(\tau)$ instead of $V(T)$.  Given a witness sequence $W = (w_1, w_2, \ldots, w_t) $ we associate with each $i \in [t] $ a witness tree $\tau_W(i) $ constructed in $i$ steps as follows:
let $\tau_W^{i}(i) $ be an isolated vertex labelled by $w_i$; then, going backwards for each $j = i-1, i-2, \ldots, 1$, if there is a vertex $v \in \tau_W^{j+1}(i)$ such that $ w_j \in D((v)) \cup \{ (v) \}$, then among those vertices  we choose the one having maximum 
distance from the root (breaking ties arbitrarily) and attach a new child vertex $u$ to $v$ that we label $w_j$ to get $\tau_W^{j}(i)$. If there is no such vertex $v$ then $\tau_W^{j+1}(i) = \tau_W^{j}(i)$. Finally, $\tau_W(i) = \tau_W^{1}(i)$.

We will say that a witness tree $\tau$ occurs in  a trajectory with witness sequence $W = (w_1, w_2, w_3, \ldots ) $, if there is $k  \ge 1$ such that $\tau_W(k) = \tau$.  Finally, we use the notation $\Pr[\cdot ]$ to refer to the probability of events in the probability space induced by the execution of the Moser-Tardos algorithm.

\begin{lemma}[The witness tree lemma~\cite{MT}]\label{witness_trees}
For  every witness tree $\tau$, $\Pr[\tau] \le \prod_{v \in V(\tau) }  \mu(A_{(v)})$.
\end{lemma}

\subsubsection{The Analysis}

Let $E_i$ be the event that the complete execution of the Moser-Tardos algorithm ever resamples query variable $x_i$ after the time, $t_i$, that it returned a response for it. Let $c_j$ be a constraint that contains $x_i$ and let $E_{i,j} \subseteq E_i$ denote the event that constraint $c_j$ is resampled after $t_i$. Clearly, $E_i \subseteq \bigcup_{ c_j \ni x_i } E_{i,j} $. The key insight is that in order for  $E_{i,j}$ to occur, it should be that at least $r$ constraints that form a path in the dependency graph which ends in $j$ must have been resampled after $t_i$. This is because, by the nature of our algorithm, right after step $t_i$, every constraint in $\mathrm{Ball}(i,r)$ is satisfied. This implies that the (first) resampling of the bad event $A_j$ that corresponds to event $E_{i,j}$ occurring  will be associated with a witness tree of size at least $r$. Thus, if $r$ is large, $E_{i,j}$ is  unlikely. Lemma~\ref{bound} makes this idea rigorous.
\begin{lemma}\label{bound}
Let $\mathcal{W}_{j,s}$ denote the set of all witness trees of size at least $s$ whose root is labelled by $j$. Then,
\begin{align*}
\sum_{ \tau \in \mathcal{W}_{j,s}} \Pr[ \tau  ] \le \psi_j (1-\epsilon)^s  \enspace.
\end{align*}
\end{lemma}

We prove Lemma~\ref{bound} in Subsection~\ref{ProofBound}. Using it, we can show the following.
\begin{lemma}\label{error_prob_radius}
Let $\eta= \max_{x \in \mathcal{X} } \sum_{ c_j \ni x} \psi_j $. If
\begin{align*}
r \ge \frac{ \log \left( (\delta -  \frac{q}{n ^{ 2} }   ) ^{-1} q \eta   \right) }{  \log \left( 1/ (1- \epsilon) \right) } \enspace ,
\end{align*}
then the probability that our algorithm answers at least one  query incorrectly  is at most $\delta - \frac{q}{n^2}$. 
\end{lemma}

\begin{proof}
Combining Lemma~\ref{bound} with our observation regarding the minimum size of witness trees related to event $E_{i,j}$,  we obtain
\begin{align*}
\Pr[ E_{i,j} ] \le \sum_{\tau \in \mathcal{W}_{j,r}}  \Pr[ \tau  ] \le \psi_j (1-\epsilon)^r \enspace.
\end{align*}
Thus, taking  $r \ge \frac{ \log \left( (\delta -  \frac{q}{n ^{ 2} }   ) ^{-1} q \eta   \right) }{  \log \left( 1/ (1- \epsilon) \right) }$ and applying the union bound we obtain
\begin{align}\label{resampling_bounds}
\Pr \left[ \bigcup_{i \in [q] } \bigcup_{ c_j \ni x_i } E_{i,j}    \right] & \le   (1- \epsilon)^r    \sum_{ i =1 }^q  \sum_{ c_j \ni x_i  } \psi_j   \nonumber \\
										              &   \le  q \eta    (1- \epsilon)^r   \le \delta   - \frac{q }{n^{ 2} }  \enspace.
\end{align}
\end{proof}

\subsubsection{Proof of Lemma~\ref{bound}}\label{ProofBound}

A typical argument used in the algorithmic LLL literature to estimate sums over sets of witness trees, such as the sum in the statement of Lemma~\ref{bound}, is to consider a Galton-Watson branching process that produces each witness tree in the set of interest (and perhaps other trees) with positive probability. The idea is to then relate the probability of the branching process generating each tree with the probability that the tree occurs in the algorithm and exploit that the sum of the probabilities in the process is, by definition, bounded by 1. 
\begin{lemma}[\cite{MT}]\label{trick}
Let $\mathcal{W}_j$ denote the set of witness trees whose root is labeled by $j$. There exists a branching process  that outputs each witness tree $\tau \in \mathcal{W}_j$ with probability
\begin{align*}
p_{\tau}  =  \psi_j^{-1} \prod_{v \in V(\tau)}  \frac{\psi_{(v)}  }{\sum_{ S \subseteq D((v) ) \cup \{ (v) \}  } \prod_{f \in S } \psi_f  }  \enspace.
\end{align*}
\end{lemma}

Observe that since $\mathcal{W}_{j,s } \subseteq \mathcal{W}_j$, Lemma~\ref{trick} implies that
\begin{align}\label{second}
\psi_j \ge \psi_j \sum_{ \tau \in \mathcal{W}_{j,s}}  p_{\tau} = 
 \sum_{ \tau \in \mathcal{W}_{j,s}} \prod_{v \in V(\tau)} \frac{\psi_{(v)}  }{\sum_{ S \subseteq D((v) ) \cup \{(v) \} } \prod_{f \in S } \psi_f  }   \enspace.
\end{align}

Lemma~\ref{witness_trees} implies~\eqref{aris_school_1} below, the fact that the LLL conditions hold with $\epsilon$ slack implies~\eqref{aris_school_2}, the fact that every witness tree in $\mathcal{W}_{j,s}$ has size at least $s$ implies~\eqref{aris_school_3}, while inequality~\eqref{second}, finally, implies~\eqref{aris_school_4}.
\begin{align}
&\sum_{ \tau \in \mathcal{W}_{j,s}} \Pr[ \tau ]   \\
\le& \sum_{ \tau \in \mathcal{W}_{j,s}} \prod_{v \in V(\tau)} \mu((v)) 
\label{aris_school_1}  \\ 
 \le & \sum_{ \tau \in \mathcal{W}_{j,s}} \prod_{v \in V(\tau)} \frac{(1-\epsilon) \psi_{(v)}  }{\sum_{ S \subseteq D((v) ) \cup \{(v) \} } \prod_{f \in S } \psi_f  } 
\label{aris_school_2} \\
\le &  (1-\epsilon)^s   \sum_{ \tau \in \mathcal{W}_{j,s}} \prod_{v \in V(\tau)} \frac{\psi_{(v)}  }{\sum_{ S \subseteq D((v) ) \cup \{(v) \} } \prod_{f \in S } \psi_f  } 
\label{aris_school_3}  \\
\le & (1-\epsilon)^s \psi_j 
\label{aris_school_4} \enspace. 
\end{align}

%
%
%
%

\subsection{Concluding the Proof}

%

Recall that $\eta= \max_{x \in \mathcal{X} } \sum_{ c_i \ni x} \psi_i $, that  $\xi = \max_{i \in [ m] } \log (1+ \psi_i )$, and that $t$ denotes the required upper bound on the running time of our algorithm on a single query. Lemma~\ref{running_time}  and Lemma~\ref{error_prob_radius} imply that there exists $C = \widetilde{O}(1) $ such that if
\begin{align}\label{interval}
r \in \left[\frac{ \log \left( (\delta -  \frac{q}{n ^{ 2} }   ) ^{-1} q \eta   \right) }{  \log \left( 1/ (1- \epsilon) \right) },   \frac{ \log \left(  \frac{ t -  s}{\xi C }  \log \frac{1}{ 1- \epsilon} \right) }{ \log (k d)}    \right] \enspace,
\end{align}
where $ s=  \frac{2 \log n }{ \log ( 1/ (1- \epsilon )  } $, then the probability that the algorithm aborts in Line~\ref{abort}  or responds inaccurately on any query is at most $ \frac{1}{n^2} + (\delta - \frac{q}{n^2} ) \le  \delta$.

Recall that $\max_{i \in [m ] } \psi_i =  O( n^{\lambda}) $ and that $\zeta = \log (1/(1-\epsilon)) / \log kd $. It is not hard to see that if $ q= n^{\alpha}, t = n^{\beta}, \delta = n^{-\gamma} $ and $  \beta \zeta > \alpha  + \gamma  + \lambda$, then  the interval in~\eqref{interval} is non-empty for large enough $n$, concluding the proof of Theorem~\ref{main}. The proof of Remark~\ref{restricted_version} is very similar.

\bibliography{kolmo}

\appendix

\section{Proofs for our  Applications} 

In this section we prove Theorems~\ref{SAT},~\ref{lca_theorem} and~\ref{non_uniform_hyper}.

\subsection{Proof of Theorem~\ref{SAT}}

We first briefly recall the application of the LLL in~\cite{2ke}. For each variable $x_i$, let $d_i$ denote the number of clauses in which $x_i$ occurs and assume that $\theta_i d_i$ of these occurrences are positive, for some $\theta_i \in [0,1]$. Let $d = \max_{ i \in [n] } d_i $ and let $\mu$ be the product measure over the variables of $\phi$ that sets each variable $x_i$ to \emph{true} with probability $\frac{1}{2} + \frac{ 2(1-\theta_i)d_i  - d}{2dk } $.  In~\cite{2ke} it is shown that if $d (k+1)  \le  2^{k+1}/\mathrm{ e}$ and we set $\psi_j = \frac{e }{ 2^k-e}  = O(1)$ for each $j  \in [m]$, then the LLL conditions are satisfied.

To establish part (a) of Theorem~\ref{SAT}, recall the definition of $\zeta$ and $\epsilon$ in Theorem~\ref{main} and notice that it implies that 
\begin{align}
1 - \epsilon := (kd)^{-\zeta} \ge ((k+1)d)^{-\zeta} \enspace.
\end{align}
Thus, in order to meet the requirement of Theorem~\ref{main} that the LLL conditions hold with an $\epsilon$-slack, i.e., that  $d   (k+1)  \mathrm{e}    \le  (1 - \epsilon) 2^{k+1}$, it is enough that
\begin{align}\label{basic_yo}
 d   (k+1)  \mathrm{e}    \le   ((k+1)d)^{-\zeta}    2^{k+1} \enspace.
\end{align}
Setting $\eta = \zeta$ in~\eqref{basic_yo}, we get the condition of part (a) of Theorem~\ref{SAT}, concluding the proof. Part (b) of Theorem~\ref{SAT} is a straightforward application of Theorem~\ref{main} and  Remark~\ref{restricted_version}.

\subsection{Proof of Theorem~\ref{lca_theorem}}\label{coloring_proof}

We'll need to briefly recall the key ideas in the analysis of the algorithm of~\cite{molloy1997bound}.

In the first phase, the algorithm operates on the set $\Omega$ of complete but not necessarily proper colorings of $G$ with at most $\Delta +1 - Z$ colors, where $Z= B/(\mathrm{e}^6 \Delta)$. For a vertex $v$ and a state $\sigma \in \Omega$, say that a color $c$ is \emph{stable} if it is assigned to at least two non-adjacent neighbors of $v$ and, moreover, all neighbors of $v$ with color $c$ do not  belong in a monochromatic edge in $\sigma$. Let $X_v(\sigma) $ be the number of stable colors for $v$ at $\sigma$. For each vertex $v$, define the bad event  $A_v = \left\{ \sigma \in \Omega:  X_v(\sigma) \le 
Z
\right\} $ with respect to the probability space $(\mu , \Omega)$, where $\mu$ is the uniform measure over $\Omega$. A coloring $\sigma^* \in \Omega$ that avoids all bad events, can be efficiently transformed to a proper coloring of $G$. To see this, consider the partial \emph{proper} coloring $\sigma'$ that results by uncoloring every vertex in $\sigma^*$ that belongs in a monochromatic edge. Since $\sigma^*$ avoided all bad events, this means that in the neighborhood of every [uncolored] vertex in $\sigma'$, at least $Z$ colors appear at least twice. Therefore, in $\sigma'$, for every vertex $v$ both of the following hold: (i) $v$ has at most $\Delta-(2Z+s(v))$ uncolored neighbors, where $s(v)$ is the number of colors appearing exactly once in the neighborhood of $v$, and (ii) at least $\Delta+1-Z-(Z+s(v))$ colors are available, i.e., do not appear in $v$'s neighborhood. Thus, the graph induced by the uncolored vertices can be colored with available colors using the greedy heuristic.

To prove that we can find efficiently a coloring $\sigma^* \in \Omega$ that avoids all bad events we use the following two lemmas from the analysis of~\cite{molloy1997bound} (slightly modified to fit our needs). Below, both the expectation and the probability are with respect to $\mu$.
\begin{lemma}[\cite{molloy1997bound}]\label{expect_coloring}
$\ex[X_v] \ge 2Z
$.
\end{lemma}
\begin{lemma}[\cite{molloy1997bound}]\label{concent_coloring}
$$\Pr\left[ |X_v -\ex[ X_v] | > (\log \Delta)^2 \sqrt{ \ex[X_v] }\right]  \le \Delta^{- \frac{\log \Delta }{1000} }.      $$
\end{lemma}
Lemmata~\ref{expect_coloring} and~\ref{concent_coloring} imply that if  $B \ge \Delta \log^4 \Delta$ and $\Delta$ is  large enough, then $\Pr[ A_v ]  \le  \Delta^{- \frac{\log \Delta}{1000}  }$. On the other hand, each bad event $A_v$ is mutually independent from all but at most $\Delta^4$ other bad events, since it only depends on the color of vertices which are joined to $v$ by a path of length at most $2$. Thus, for large enough $\Delta$, if $\psi_i = \frac{1}{ \Delta^4-1}$ for every $i \in [m]$, then the LLL condition is satisfied,  implying that the Moser-Tardos algorithm finds a coloring that avoids all bad events quickly.

\begin{proof}[Proof of Theorem~\ref{lca_theorem}]

We consider the constraint satisfaction problem with one variable and one constraint per vertex $v$ of the graph,  the variable expressing the color of $v$, and the constraint including all vertices (variables) within distance 2 from $v$ and forbidding all joint value assignments for which $A_v$ occurs. Observe that knowing the color of the vertices included by the constraint of $v$ is (more than) enough information to determine if $v$ belongs in any monochromatic edge and, thus, whether it retains its color when we uncolor all vertices belonging in monochromatic edges. With this in mind, our local algorithm is the following.

Let $q$ be the number of queries. For each $i \in [q]$, to answer the $i$-th query we use the procedure described in the proof of Theorem~\ref{main} to satisfy all constraints within a ball of some radius of the queried vertex $v_i$. Naturally, this colors all vertices in the neighborhood of $v_i$ (and probably many others). If, in the resulting coloring, we find that $v_i$ participates in a monochromatic edge we ``guess" that $v_i$ will be uncolored at the end of the first phase, otherwise we ``guess" that it will have the colored assigned by our procedure. In the latter case, we return this color as our answer. To answer queries for vertices that we guess will be uncolored at the end of the first phase, we simulate the greedy coloring of the second phase, using the ordering of the vertices by the queries. That is, whenever we guess that $v$ is uncolored, we chose one of its available colors, $c$, return it as our answer to the query, and record $c$ as the color of $v$. If we later need to answer a query for a neighbor $v'$ of $v$ that we also guess to be uncolored, we do not consider $c$ an available color for $v'$. Thus, if our algorithm does not make any wrong guesses, it doesn't make any error at all.



To bound the probability of error we use Theorem~\ref{main}. Since the constraints correspond to the events $A_v$, we see that $k =  \Delta$ and $d = \Delta^4$ . We are interested in responding to at most $n$ queries, i.e.,  $\alpha = 1$.  Letting $p = \max_{v \in V} \Pr[ A_v] $ and setting $\psi_i = \frac{1}{ \Delta^4 -1}$ we see that the LLL condition is satisfied if $p  \Delta^4 \mathrm{e} \le 1 $. Since $\beta, \gamma$ are constants and $p \le \Delta^{- \frac{ \log \Delta }{1000 }  } $, we see that $ p \Delta^4 \mathrm{e}  \le \Delta^{- ( \frac{ 5(1 + \gamma)}{ \beta } +1)  }$ for large enough $\Delta$. Letting $\Delta^{- ( \frac{ 5(1 + \gamma)}{ \beta } +1)  }  =: 1- \epsilon $ and noting that $\max_{i \in [m]} = O(1)$, i.e., $\lambda =0$, we obtain
\begin{align*}
\zeta = \frac{ \log (\frac{1}{1-\epsilon} )}{ \log (kd) } = \frac{ \left( \frac{5(1+\gamma)}{\beta} + 1\right) \log \Delta  }{ 5 \log \Delta }    > \frac{1+\gamma}{ \beta}  \enspace,
\end{align*}
and in turn that  $\beta \zeta > 1+  \gamma + 0$. Thus, Theorem~\ref{main} applies, concluding the proof.
\end{proof}

\subsection{Proof of Theorem~\ref{non_uniform_hyper}}

We consider the uniform measure over all possible $2$-colorings of $\mathcal{H}$ and define one bad event, $A_e$, for each edge $e$, corresponding to $e$ being monochromatic. Clearly, $\Pr[A_e] =  \frac{1}{2^{ |e|-1} } $. If we set $\psi_e = \frac{2 x_e }{1 - x_e } $, where $x_e = \left(\frac{1}{2}\right)^{ \frac{1}{2} ( |e|-1 ) }  $, the LLL conditions are satisfied assuming that
\begin{align}\label{cond_proof}
\sum_{i \ge 3} \Delta_i 2^{-i/2 } \le  \frac{1}{ 6 \sqrt{2} } \enspace.
\end{align}
(For more details, see the proof of Theorem 19.2 in~\cite{mike_book}). Now, since $\max_{e} |e| $ and $\max_i \Delta_i$ are constants we see that the condition of Theorem~\ref{non_uniform_hyper} implies that the LLL condition holds with $\epsilon$ slack and, thus, the proof follows directly from Theorem~\ref{main} and Remark~\ref{restricted_version}.

\section{Improved LLL Criteria and Commutative Algorithms}\label{improvedLLL}

Our techniques can be generalized in two distinct directions. First, so that they apply under more permissive LLL conditions such as  the \emph{cluster expansion condition}~\cite{Bissacot} and  \emph{Shearer's condition}~\cite{Shearer}. Second, they can be used to design local computation algorithms that simulate algorithms in  the abstract settings of the algorithmic Lov\'{a}sz Local Lemma~\cite{AIJACM,HV,AIK}, where the probability space does not necessarily correspond to a product measure, and which capture the \emph{lopsided} version of the LLL~\cite{LopsTrav}. We briefly discuss these extensions below. 

Given a dependency graph $G$ over $[m]$ and a set $S \subseteq [m] $ we denote by $\mathrm{Ind}([m] ) =  \mathrm{Ind}_G([m] )$ the family of subsets of $S$ that correspond to independent sets in $G$.
\paragraph*{Cluster Expansion condition.}
The cluster expansion criterion strictly improves upon the General LLL criterion~\eqref{eq:LLL} by taking advantage of the local density of the dependency  graph. 
\begin{definition}
Given a sequence of positive real numbers $\{ \psi_i \}_{i=1}^{m}$, we say that the cluster expansion condition is satisfied if for each $i \in [m]$
\begin{align*}
\frac{\mu(A_i) }{\psi_i }  \sum_{ S \in \mathrm{Ind}( D(i) \cup \{i \} ) } \prod_{j \in S} \psi_j \le 1 \enspace.
\end{align*}
\end{definition}
\paragraph*{Shearer's condition.} 
Shearer's condition  improves upon the general and cluster expansion LLL conditions by exploiting the global structure of the dependency graph. It is best possible in the sense that if it is not satisfied, then one can always construct a probability space and bad events that are compatible with the given dependency graph, for which the probability of avoiding all bad events is zero.
\begin{definition}\label{ShearerLLL}
Let $\mu =( \mu_1, \mu_2, \ldots, \mu_m) \in \mathbb{R}^{m}  $ be the real vector such that $\mu_i = \mu(A_i)$.  For $S \subseteq [m]$ define  $\mu_S = \prod_{i \in S } \mu_i$ and the polynomial $q_S$
\begin{align*}
q_S  = q_S(\mu) = \sum_{ \substack{ I  \in \mathrm{Ind}([m])   \\ S \subseteq I} } (-1)^{ |I|- |S| }\mu_{I}   \enspace.
\end{align*}
We say that the Shearer's condition is satisfied if $q_S(\mu) \ge 0 $ for all $S \subseteq [m]$, and $q_{\emptyset}(\mu) > 0 $. 
\end{definition}

For the variable setting, the statement of our results remain identical under the cluster expansion and, essentially identical, under Shearer's conditions ($\psi_i$ is replaced by $ q_{\{i \}}( \mu) / q_{ \emptyset }(\mu) $ and we say that the condition holds with $\epsilon$-slack for a given vector $\mu$, if it simply holds for vector $(1+\epsilon) \mu$.) The only thing that changes in the analysis is the bound for the sum of probabilities of witness trees of large size in Lemma~\ref{trick}. (We refer the reader to Section 4 in~\cite{LLLWTL} for further details.)

The first result that made the LLL constructive in a non-product probability space was due to Harris and Srinivasan in~\cite{PermHarris}, who considered the space of permutations endowed  with the uniform measure. Subsequent works by Achlioptas and Iliopoulos~\cite{AIJACM,AIK,AIS} introducing the flaws/actions framework, and of Harvey and Vondr\'{a}k~\cite{HV} introducing the resampling oracles framework, made the LLL constructive in more general settings. These frameworks~\cite{AIJACM,AIK,HV,AIS} provide tools for analyzing \emph{focused} stochastic search algorithms~\cite{papafocus}, i.e., algorithms which, like the Moser-Tardos algorithm, search by repeatedly selecting a flaw of the current state and moving to a random nearby state that avoids it, in the hope that, more often than not, more flaws are removed than introduced, so that a flawless object is eventually reached.

Our techniques can be extended to  these more general settings assuming they are \emph{commutative}, a notion introduced by Kolmogorov~\cite{Kolmofocs,AIS}. While we will not define the class of commutative algorithms here for the sake of brevity, we note that it contains the vast majority of LLL algorithms, including the Moser-Tardos algorithm. The reason why our results apply in this case is because the \emph{witness tree lemma}, i.e., Lemma~\ref{witness_trees}  for the case of the Moser-Tardos algorithm (which was key to our analysis) holds for commutative algorithms~\cite{LLLWTL,AIS}.

\end{document}